\documentclass[11pt,reqno]{amsart}
\usepackage{geometry}
\usepackage{graphicx}
\usepackage{amsmath}
\usepackage{amssymb}
\usepackage{amsaddr}
\usepackage{epstopdf}
\usepackage{color}
\usepackage{amsthm}
\usepackage{diagbox}
\usepackage{multirow}
\usepackage{hyperref}
\usepackage[section]{algorithm}
\usepackage{float}
\usepackage[round, sort]{natbib}
\usepackage{caption}

\makeatletter
\newenvironment{breakablealgorithm}
  {
   \begin{center}
     \refstepcounter{algorithm}
     \hrule height.8pt depth0pt \kern2pt
     \renewcommand{\caption}[2][\relax]{
       {\raggedright\textbf{\fname@algorithm~\thealgorithm} ##2\par}%
       \ifx\relax##1\relax 
         \addcontentsline{loa}{algorithm}{\protect\numberline{\thealgorithm}##2}%
       \else 
         \addcontentsline{loa}{algorithm}{\protect\numberline{\thealgorithm}##1}%
       \fi
       \kern2pt\hrule\kern2pt
     }
  }{
     \kern2pt\hrule\relax
   \end{center}
  }
\makeatother

\newtheorem{thm}{Theorem}[section]

\title[Testing similarity of parametric competing risks models]{Testing similarity of parametric competing risks models for identifying potentially similar pathways in healthcare}

\author{Kathrin M\"ollenhoff}
\address{Institute of Medical Statistics and Computational Biology, Faculty of Medicine, University of Cologne, Germany; \\correspondence to kathrin.moellenhoff@uni-koeln.de}

\author{Nadine Binder}
\address{Institute of General Practice/Family Medicine, Medical Center and Faculty of Medicine, University of Freiburg, Germany}

\author{Holger Dette}
\address{Department of Mathematics, Ruhr-Universität Bochum, Germany}

\date{\today}                                    

\begin{document}

\begin{abstract}
The identification of similar patient pathways is a crucial task in healthcare analytics. A flexible  tool to address this issue are parametric competing risks models, where transition intensities may be specified by a variety of parametric distributions, thus in particular being possibly time-dependent. 
We assess the similarity between two such models by examining the transitions between different health states. This research introduces a method to measure the maximum differences in transition intensities over time, leading to the development of a test procedure for assessing similarity. We propose a parametric bootstrap approach for this purpose and provide a proof to confirm the validity of this procedure. The performance of our proposed method is evaluated through a simulation study, considering a range of sample sizes, differing amounts of censoring, and various thresholds for similarity. Finally, we demonstrate the practical application of our approach with a case study from urological clinical routine practice, which inspired this research.
\end{abstract}

\maketitle

\section{Introduction}\label{intro}

In the evolving landscape of healthcare analytics, the quest to identify similar patient pathways through various treatments and diagnostics is crucial. This paper focuses on a pivotal aspect of healthcare research: the utilization of flexible parametric competing risks models to test for similar treatment pathways across different patient populations.

Competing risks models, a special case of multi-state models (see, e.g., \citet{andersen_competing_2002}, \citet{andersen_multi-state_2002}), offer a sophisticated means to dissect and understand the intricacies of patient healthcare journeys. These models not only track transitions between different health states but also allow for a nuanced analysis of whether different treatment steps still lead to similar subsequent transitions. This research seeks to leverage these models to test for similarities in healthcare pathways, with the overarching goal of enhancing clinical decision-making. In this regard, we are particularly interested in deciding whether two competing risks models can be assumed to be \textit{similar}, or, in other words, \textit{equivalent}. Once similarity has been established, clinical decision making can profit a lot of this knowledge. For example, from the perspective of a clinical practitioner, it is important to know with regard to further decisions, whether healthcare pathways following two different initial treatments are similar or not, taking all possibly occurring events into account.  Moreover, if two pathways turn out to be equivalent, data can be pooled for further common and hence more reliable analyses.

Assessing similarity of competing risks in multi-state models has rarely been addressed in the literature to date. For the simplest case, the classical two-state survival model, several methods are available. The traditional approach of an equivalence test in this scenario is based on an extension of a log-rank test and assumes a constant hazard ratio between the two groups \cite{wellek1993}. However, this assumption, which is rarely assessed and often violated in practice as indicated by crossing survival curves \citep{li2015,jachno2019}, has been generally criticized \citep{hernan2010,uno2014}. As an alternative, \citet{com1993} introduce a non-parametric method, based on the difference of the survival functions and without assuming proportional hazards. In addition, a parametric alternative has recently been proposed by \citet{moellenhoff2023}, who consider a similar test statistic, but assume parametric distributions for the survival and the censoring times, respectively.

Recently, \citet{binder2022similarity} extend the considerations on similarity testing to competing risks models by introducing a parametric approach based on a bootstrap technique introduced earlier \cite{detmolvolbre2015}. They propose performing individual tests for each transition and conclude equivalence for the whole competing risks model if all individual null hypotheses can be rejected, according to the intersection union principle (IUP) \cite{berger1982}. Their approach, while effective, has some areas for improvement. First, with an increasing number of states the power decreases substantially, as the IUP is rather conservative \cite{phillips1990}.
Secondly, their approach builds on the assumption of constant transition intensities, i.e., exponentially distributed transition times, which can sometimes be to simplistic (as discussed in, e.g., works by \citet{hill2021} and \citet{von2017}). Therefore, exploring more flexible methods will typically offer a more fitting model for the underlying data.

The method presented in this paper improves both of these aspects. First, it allows for any parametric model, meaning in particular time-dependent transition intensities, and these parametric distributions can vary across transitions, resulting in a very flexible modeling framework. 
Second, we propose another test statistic, which results in one global test instead of combining individual tests for each state and thus results in higher power. The paper is structured as follows.  In Section 2, we define the modelling setting, outline the algorithmic procedure for testing the global hypotheses, and provide a corresponding proof of the new test procedure. In Section 3 we demonstrate the validity of the new approach and compare its performance to the previous method \cite{binder2022similarity}. Finally, in Section 4, we explain the application example which inspired this research. Thereby we particularly highlight the need to consider flexible parametric models whose specific estimators motivate further evaluations of the new method. Finally, we close with a discussion.

\section{Methods} \label{msm}

\subsection{Competing risk models and parameter estimation} \label{crm}

Following \citet{andersen_statistical_1993}, we consider two independent Markov processes \begin{equation}
\label{model}
(X^{(\ell )}(t))_{t \geq 0} \quad\quad (\ell  =1,2)
\end{equation}
with state spaces $\{0, 1, \ldots, k\}$ to model the event histories as competing risks for samples of two different populations $\ell = 1,2$. 
The processes have possible transitions from state $0$
to state $j \in \{1, \ldots, k\}$ with transition probabilities 

\begin{equation}
\mathbb{P}^{(\ell  )}_{0j}(0,t) = \mathbb{P}(X^{(\ell  )}(t) = j | X^{(\ell  )}(0) = 0).
\end{equation}

Every individual starts in state $0$ at time $0$, i.e. $P(X(0)=0)=1$. The time-to-first-event is defined as stopping time $T = \inf\{t > 0 \mid X(t) \neq 0\}$ and the type of the first event is $X(T) \in {1,\ldots, k}$.  
The event times can possibly be right-censored, so that only the censoring time is known, but no transition to another state could be observed. In general, we assume that censoring times $C$ are independent of the event times $T$.
Let

\begin{equation} \label{intensities}
\alpha_{0j}^{(\ell  )}(t) = \lim_{\Delta t \rightarrow 0} \frac{\mathbb{P}^{(\ell  )}_{0j}(t,t + \Delta t)}{\Delta t} 
\end{equation}
denote the cause-specific transition intensity from state $0$ to state $j$ 
for the  $\ell$th model. 
The transition intensities, also known as cause-specific hazards, completely determine the stochastic behavior of the process.
Specifically, $\mathbb{P}^{(\ell  )}_{00}(0,t)=\exp\big(-\sum_{j=1}^k \int_0^t  \alpha^{(\ell  )}_{0j}(u)du\big)=\mathbb{P}(T^{(\ell  )}\geq t)=S^{(\ell  )}(t)$ denotes the marginal survival probability, that is the probability of not experiencing any of the $k$ events prior to time point $t$.

We here consider parametric models for the intensities, that is $\alpha_{0j}^{(\ell  )}(t)=\alpha_{0j}^{(\ell  )}(t, \theta^{(\ell)}_{0j})$,
where 
\begin{align}\label{def:theta}
\theta^{(\ell)}_{0j}=(\theta^{(\ell)}_{0j1},\ldots,\theta^{(\ell)}_{0jp_{j}})^\top 
\end{align}
denotes a  $p_j$-dimensional parameter vector specifying the underlying distribution. Typical examples of parametric event-time models are given by the exponential, the Weibull, the Gompertz  or the log-normal distribution, just to mention a few (see e.g. \citet{kalbfleisch2011statistical}). Except for the exponential distribution, the intensities vary over time which makes the estimation procedure more complex compared to the situation of constant intensities.
 
For deriving the likelihood function to obtain estimates $\hat\theta^{(\ell)}_{0j}$ of the parameters  in \eqref{def:theta}, 
we consider possibly right-censored event times of individuals and assume that two independent samples $X^{(1)}_1, \ldots, X^{(1)}_{n_1} $ and  $X^{(2)}_1, \ldots, X^{(2)}_{n_2}$ from  Markov processes \eqref{model}
are observed over the interval $\mathcal{T}=[0,\tau]$, each containing the state 
and transition time 
(or the censoring time, respectively) of an individual $i$ in group $\ell$. 
Thus we observe $X_i^{(\ell  )}=(\tilde{T}_i^{(\ell)},X^{(\ell)}(\tilde{T}_i^{(\ell)}))$, where $ \tilde{T}_i^{(\ell)}=\min(T_i^{(\ell)},C_i^{(\ell)})$, 
$i=1,\ldots,n_\ell$. The total number of individuals is given by $n:=n_1+n_2$.

Following \citet{andersen2002competing}, in case of type I censoring, i.e. a fixed end of the study given by $\tau$, each individual $i$ contributes a factor to the likelihood function given by $S(C_i)$, 
whereas if there was a transition to state $j$ at time $T_i$ the factor would be $S(T_i)\alpha_{0j}(T_i,\theta_{0j})$ (group index $\ell$ omitted here).
Consequently the corresponding likelihood function in the $\ell$th group, based on $n_\ell$ independent observations, is given by the product
\begin{equation}\label{likelihood}
\mathcal{L}_\ell (\theta^{(\ell)})=\prod_{i=1}^{n_\ell} S^{(\ell  )}(\tilde{T}_i^{(\ell)}) \prod_{j=1}^k \alpha^{(\ell)}_{0j}(\tilde{T}_i^{(\ell)},\theta^{(\ell)}_{0j})^{I\{X^{(\ell)}(\tilde{T}_i^{(\ell)})=j\}},
\end{equation}
where 
\begin{equation}
\label{hol2}
\theta^{(\ell)}=((\theta^{(\ell)}_{01})^\top,\ldots,(\theta^{(\ell)}_{0k})^\top)^\top
\end{equation}
is the $p := \sum_{j=1}^k p_j$-dimensional parameter  vector specifying the underlying distributions and hence the transition intensities $\alpha^{(\ell)}_{0j}(t)$. 
As $\tilde{T}_i^{(\ell)}=T_i^{(\ell)}$, if individual $i$ had a transition to any of the $k$ states, we get, taking the logarithm of \eqref{likelihood},

\begin{equation}\label{loglikelihood}
\log \mathcal{L}_\ell (\theta^{(\ell)})=\sum_{i=1}^{n_\ell} \log(S^{(\ell  )}(\tilde{T}_i^{(\ell)}))+ \sum_{i=1}^{n_\ell}\sum_{j=1}^k I\{X^{(\ell)}(T_i^{(\ell)})=j\}\log(\alpha^{(\ell)}_{0j}(T_i^{(\ell)},\theta^{(\ell)}_{0j})).
\end{equation}
By maximizing the functions $\log \mathcal{L}_1$ and $\log \mathcal{L}_2$ in \eqref{loglikelihood} we obtain ML estimates  $\hat\theta^{(1)}$ and $\hat\theta^{(2)}$, respectively.

In case of random right-censoring, we assume that the censoring times $C$ 
follow a particular distribution with density $g=g(t,\psi)$ and distribution function $G=G(t,\psi)$, where $\psi$ denotes the parameter specifying the censoring distribution. Technically, assuming random right-censoring is incorporated in the likelihood as adding an additional state to the model. Precisely, if an individual $i$ is censored at censoring time $C_i$, the contribution to the likelihood is given by $\mathbb{P}(\tilde{T}_i=C_i,X(\tilde{T}_i)=0)=\mathbb{P}(\tilde{T}_i=C_i,T_i>C_i)=S(C_i)\cdot g(C_i)$ and thus the likelihood in \eqref{likelihood} is extended by an additional factor and, in group $\ell$, becomes 

\begin{equation}\label{likelihood_cens}
\mathcal{L}_\ell (\theta^{(\ell)},\psi^{(\ell)})=\prod_{i=1}^{n_\ell} S^{(\ell  )}(\tilde{T}_i^{(\ell)}) g^{(\ell)}(\tilde{T}_i^{(\ell)},\psi^{(\ell)})^{I\{X^{(\ell)}(\tilde{T}_i^{(\ell)})=0\}} \prod_{j=1}^k \alpha^{(\ell)}_{0j}(\tilde{T}_i^{(\ell)},\theta^{(\ell)}_{0j})^{I\{X^{(\ell)}(\tilde{T}_i^{(\ell)})=j\}},
\end{equation}

and, accordingly, the log-likelihood in \eqref{loglikelihood} becomes

\begin{align}\label{loglikelihood_cens}
\log \mathcal{L}_\ell (\theta^{(\ell)},\psi^{(\ell)})&=\sum_{i=1}^{n_\ell} \log(S^{(\ell  )}(\tilde{T}_i^{(\ell)}))+ \sum_{i=1}^{n_\ell}  I\{X^{(\ell)}(\tilde{T}_i^{(\ell)})=0\}\log g^{(\ell)}(\tilde{T}_i^{(\ell)},\psi^{(\ell)}) \nonumber\\
&+\sum_{i=1}^{n_\ell}\sum_{j=1}^k I\{X^{(\ell)}(T_i^{(\ell)})=j\}\log(\alpha^{(\ell)}_{0j}(T_i^{(\ell)},\theta^{(\ell)}_{0j})).
\end{align}

\bigskip

\subsection{Similarity of competing risk models} \label{simapp}


An intuitive way to define  similar competing risk models is by measuring the maximum distance between transition intensities and decide for similarity if this distance is small. Note that, due to an easier readability, we omit the dependency of the intensities $\alpha^{(\ell)}_{0j}$ on  the parameters $\theta^{(\ell)}_{0j}$, $j=1,\ldots k$, throughout the following discussion. Therefore the hypotheses are given by
\begin{equation}\label{h2}
H_0 : \mbox{ there exists an index } j \in \{ 1,\ldots,k \} \mbox{ such that } \| \alpha^{(1)}_{0j} - \alpha^{(2)}_{0j} \|_\infty \geq \Delta
\end{equation}
versus
\begin{equation}\label{h3}
H_1 : \mbox{ for all } j \in \{ 1,\ldots,k \} \qquad  \| \alpha^{(1)}_{0j} - \alpha^{(2)}_{0j} \|_\infty < \Delta,
\end{equation}
where $\Delta$ is a pre-specified threshold and $\| f-g\|_\infty = \sup_{t \in \mathcal{T}} \mid f(t) - g(t) \mid$ denotes the maximal deviation
between the functions $f$ and $g$.
\bigskip

This test problem can be addressed by two different types of test procedures. 
If one is interested in comparing each pair of transition intensities $\alpha^{(1)}_{0j}(t)$ and $\alpha^{(2)}_{0j}(t)$, $j=1,\ldots,k$, over the entire interval $[0,\mathcal{T}]$ individually, we propose to do a separate test for each of these $k$ comparisons and to combine them via IUP \cite{berger1982} as described in Binder et al. \cite{binder2022similarity}. 
This method has the advantage that one can make inference about particular differences between transitions and the threshold in \eqref{h3} can be replaced by individually chosen thresholds $\Delta_j$, $j=1,\ldots,k$, for each single comparison. 
However, if the threshold $\Delta$ is globally chosen, as stated in \eqref{h2} and \eqref{h3}, applying the same principle means that the similarity of the $j$-th transition intensities is assessed by testing the individual hypothesis
\begin{equation}\label{h2ind}
H^j_0 : \| \alpha^{(1)}_{0j} - \alpha^{(2)}_{0j} \|_\infty \geq \Delta
\end{equation}
versus
\begin{equation}\label{h3ind}
H^j_1 : \| \alpha^{(1)}_{0j} - \alpha^{(2)}_{0j} \|_\infty < \Delta.
\end{equation}
However, combining these individual tests to obtain a global test decision results in a noticeable loss of power, which is a well known consequence of tests based on the IUP \cite{phillips1990}. 
Therefore, if one is interested in claiming similarity of the whole competing risks models rather than comparing particular transition intensities, another test procedure should be considered.
This procedure is based on re-formulating $H_1$ in \eqref{h3} to 
\begin{equation}\label{h1new}
 H_1:\ \max_{j=1}^k \| \alpha^{(1)}_{0j} - \alpha^{(2)}_{0j} \|_\infty < \Delta,\end{equation}
which gives rise to another test statistic.
Based on this, the following algorithm describes a much more powerful procedure for testing the hypotheses  \eqref{h2} against \eqref{h1new}. 
\bigskip
\begin{breakablealgorithm} 
\caption{}
\label{alg1} 
\begin{itemize}
\item[(1)]	For both samples, calculate the MLE $\hat\theta^{(\ell)}$ and $\hat\psi^{(\ell)}$, $\ell=1,2$, by maximizing the log-likelihood given in \eqref{loglikelihood_cens}, in order to obtain the transition intensities $\hat\alpha^{(1)}$ and $\hat\alpha^{(2)}$ with $\hat\alpha^{(\ell)}=(\hat\alpha^{(\ell)}_{01},\ldots,\hat\alpha^{(\ell)}_{0k})$ and the parameters $\hat\psi^{(\ell)}$, $\ell=1,2$, of the underlying censoring distributions.
Note that, in case of no random censoring, it suffices to maximize the log-likelihood in \eqref{loglikelihood}. 
From the estimates, calculate the corresponding test statistic  
		$$\hat d:=\max_{j=1}^k \| \hat\alpha^{(1)}_{0j} - \hat\alpha^{(2)}_{0j} \|_\infty.$$
\item[(2)] 
In a second estimation step, we define  constrained estimates  $\overline \theta^{(1)}$ and $\overline \theta^{(2)}$ of  $\theta^{(1)}$ and $\theta^{(2)}$, 
maximizing the sum
			$\log \mathcal{L}_1 (\theta^{(1)})+\log \mathcal{L}_2 (\theta^{(2)})$
			of the log-likelihood functions defined 
			in \eqref{loglikelihood}   under the additional constraint 
			\begin{equation}\label{constr}
	\max_{j=1}^k  \| \alpha^{(1)}_{0j}-\alpha^{(2)}_{0j} \|_\infty =\Delta.
			\end{equation}
   Further define
			\begin{equation}\label{MLcons}
			\hat{\hat{\theta}}^{(\ell)}= \left\{
			\begin{array} {ccc}
			\hat \theta^{(\ell)} & \mbox{if} & \hat d \geq \Delta \\
			\overline \theta^{(\ell)} & \mbox{if} & \hat d < \Delta
			\end{array}  \right.,\   \ell=1,2,
			\end{equation}
			where $\hat{\hat{\theta}}^{(\ell)}  = ( \hat{\hat{ \theta}}^{(\ell)}_{01}, \ldots ,\hat{ \hat{ \theta}}^{(\ell)}_{0k})^\top$. 
From this we obtain constrained estimates of the transition intensities $\hat{\hat{\alpha}}_{0j}^{(\ell  )}(t)=\alpha_{0j}^{(\ell  )}(t, \hat{\hat{\theta}}^{(\ell)}_{0j})$, $j=1,\ldots,k$, $\ell=1,2$. 
Finally, note that this constraint optimization does not affect the estimation of the censoring distribution.
\item[(3)] 
By using the constrained estimates $\hat{\hat{ \alpha}}^{(\ell)}=\hat{\hat{ \alpha}}^{(\ell)}(t)=(\hat{\hat{\alpha}}_{01}^{(\ell  )}(t),\ldots,\hat{\hat{\alpha}}_{0k}^{(\ell  )}(t))$, simulate bootstrap event times $T^{*(1)}_1, \ldots, T^{*(1)}_{n_1}$ and  $T^{*(2)}_1, \ldots, T^{*(2)}_{n_2}$. Specifically we use the simulation approach as described in \citet{beyersmann2009}, where at first for all individuals survival times are simulated with all-cause hazard $\sum_{j=1}^{k}\hat{\hat{ \alpha}}^{(\ell)}_{0j}(t)$ as a function of time and then a multinomial experiment is run for each survival time $T$ which decides on state $j$ with probability $\hat{\hat{ \alpha}}^{(\ell)}_{0j}(T)/\sum_{j=1}^{k}\hat{\hat{ \alpha}}^{(\ell)}_{0j}(T)$.
In order to represent the censoring adequately, we now use the parameters $\hat\psi^{(\ell)}$, $\ell=1,2$ from step (i) to additionally generate bootstrap censoring times $C^{*(1)}_1, \ldots, C^{*(1)}_{n_1}$ and  $C^{*(2)}_1, \ldots, C^{*(2)}_{n_2}$, according to a distribution with distribution function $G^{(1)}(t,\psi^{(1)})$ and $G^{(2)}(t,\psi^{(2)})$, respectively. Finally, the bootstrap samples are obtained by taking the minimum of these times in each case, that is 
$\tilde{T}^{*(\ell)}_i=\min(T^{*(\ell)}_i,C^{*(\ell)}_i)$.
Note that, in case of no random but administrative censoring with a fixed end of the study $\tau$, we take $\tilde{T}^{*(\ell)}_i=\min(T^{*(\ell)}_j,\tau)$, $i=1,\ldots,n_{\ell},\ \ell=1,2$.

For the datasets $X^{*(1)}_1, \ldots, X^{*(1)}_{n_1} $
			and  $X^{*(2)}_1, \ldots, X^{*(2)}_{n_2}$, consisting of the potentially censored event time and the simulated state of an individual, calculate
			the MLE $\hat\alpha^{*(1)}$ and $\hat\alpha^{*(2)}$ by maximizing \eqref{loglikelihood} and the test statistic as in Step (i), that is
			$$\hat d^{*}:=\max_{j=1}^k\| \hat\alpha^{*(1)}_{0j} - \hat\alpha^{*(2)}_{0j}\|_\infty.$$

\item[(4)] Repeat Step (3) $B$ times to generate $B$ replicates of the test statistic $\hat d^{*(1)}, \ldots,  \hat d^{*(B)}$, yielding an estimate of the $\alpha$-quantile of the distribution of the statistic $d^*$, which is denoted by $ q^*_\alpha$. Finally reject the null hypothesis in \eqref{h2} if
\begin{align}
\label{hd1}
\hat{d}\leq q_\alpha^*.
\end{align}
Alternatively, a test decision can be made based on the $p$-value 
$$\hat F_{B}(\hat d) = {\frac{1}{B}} \sum_{i=1}^{B}  I\{ \hat d^{*(i)} \leq \hat d \},$$
where $\hat F_{B}$ denotes the empirical distribution function of the bootstrap sample. 
		Finally, we reject the null hypothesis \eqref{h2}  if $\hat F_{B}(\hat d)<\alpha$ for a pre-specified significance level $\alpha$. 
\end{itemize}
\end{breakablealgorithm}

\bigskip

The following result shows  that Algorithm 
\ref{alg1} defines a valid statistical test for the hypotheses \eqref{h2} and \eqref{h1new}.
The proof is deferred  to the appendix.

\begin{thm} \label{thmneu} 
Assume that $\lim_{n_1,n_2 \to \infty  } {n_1 \over n_2} = c > 0$ and 
that Assumption A -D in \cite{borgan1984} are satisfied. Further let
$$
\| f \|_{\infty , \infty }   :=   \max_{j \in \{ 1, \ldots , k \} } \| f_{j} (t) \|_{\infty } 
=
 \max_{j \in \{ 1, \ldots , k \}  \times \mathcal{T}} | f_{j}(t)| 
$$ 
denote the $\ell^{\infty}$-norm  on the set of  functions $(j,t) \to f_{j} (t)$ defined on $  \{ 1, \ldots , k \}  \times \mathcal{T}.$
Then the test defined by \eqref{hd1} is consistent and has asymptotic level $\alpha $ for the hypotheses 
\eqref{h2} and \eqref{h1new}. More precisely, 
\begin{itemize}
\item[(1)]
if  the null hypothesis in \eqref{h2} is satisfied, then we have for any $\alpha \in (0, 0.5)$
\begin{equation}\label{level2.1}
\limsup_{n \rightarrow\infty}\mathbb{P}\big(
  \hat{d}\leq q_\alpha^*
\big)\leq\alpha,
\end{equation}
\item[(2)] if the null hypothesis in \eqref{h2} is satisfied and
the set 
\begin{align} \label{det1}
 \mathcal{E} =    \Big \{ (j,t) \in \{ 1, \ldots , k \} \times \mathcal{T} ~:~ 
| \hat\alpha^{(1)}_{0j} (t) - \hat\alpha^{(2)}_{0j} (t) |  = 
 \| \hat\alpha^{(1)} - \hat\alpha^{(2)}\|_{\infty,\infty} 
\Big \} 
\end{align}
consists of one point, then we have for any $\alpha \in (0, 0.5)$
\begin{equation}\label{level2}
\lim_{n\rightarrow\infty} \mathbb{P}
\big(
   \hat{d}\leq q_\alpha^*  \big) = \left\{
\begin{array} {ccc}
0 & \mbox{if} & \max_{j=1}^k \| \hat\alpha^{(1)}_{0j} - \hat\alpha^{(2)}_{0j} \|_\infty   > \Delta  \\
\alpha & \mbox{if} & \max_{j=1}^k \| \hat\alpha^{(1)}_{0j} - \hat\alpha^{(2)}_{0j} \|_\infty  = \Delta  
\end{array} \right .,
\end{equation}
\item[(3)] if the alternative in \eqref{h1new}   is satisfied, then we have for any $\alpha \in (0, 0.5)$
\begin{equation}
\lim_{n\rightarrow\infty} \mathbb{P}
\big(
   \hat{d}\leq q_\alpha^*  \big) 
   =1.\label{consistence2}
\end{equation}
\end{itemize}
\end{thm}

\section{Simulation study}\label{sim}

\subsection{Design}
In this section we present numerous simulation results obtained by the method proposed in Algorithm \ref{alg1}. 
We assume two different settings for the distributions of the transition intensities, resulting in four different scenarios in total. All scenarios are driven by the application example given in Section \ref{appl}.
In Scenario 1 and Scenario 2 we assume the event times to follow an exponential distribution, i.e. all transition intensities are assumed to be constant. This setting is the same as already considered for the simulations in \cite{binder2022similarity}. We denote the approach mentioned therein by "Individual Method" throughout the rest of this paper, as it is based on combining three individual tests, one for each state. Consequently, in this setting all results from the two methods are directly comparable. The parameters of the constant transition intensities are given in Table \ref{tab:applparameter} in Section \ref{appl}, these are used for Scenario 1, yielding $$d=\max_{j=1}^3 \| \alpha^{(1)}_{0j} - \alpha^{(2)}_{0j} \|_\infty=\max\{0.0002,0.0006,0.0005\}=0.0006$$
for Scenario 1. For Scenario 2 we choose identical models, that is $\alpha^{(1)}_{01}=\alpha^{(2)}_{01}=0.001,\ \alpha^{(1)}_{02}=\alpha^{(2)}_{02}=0.0011$ and $\alpha^{(1)}_{03}=\alpha^{(2)}_{03}=0.0004$, respectively, resulting in a difference of $0$ for all transition intensities. \\

For the second setting, i.e. Scenario 3 and Scenario 4, respectively, we assume a Gompertz distribution for the first two states and a Weibull distribution for the third state, i.e. the intensities of the first two states are given by
\begin{equation}\label{scenario2a}\alpha^{(\ell)}_{0j}(t,\theta^{(\ell)}_{0j})=\theta^{(\ell)}_{0j1}\cdot\exp{(\theta^{(\ell)}_{0j2}\cdot t)},\ j=1,2,\ \ell=1,2,\end{equation}
where $\theta^{(\ell)}_{0j1}$ denotes the scale and $\theta^{(\ell)}_{0j2}$ the shape parameter, respectively, and the transition intensity for the third state is given by
\begin{equation}\label{scenario2b}
\alpha^{(\ell)}_{03}(t,\theta^{(\ell)}_{03})=\frac{\theta^{(\ell)}_{032}}{\theta^{(\ell)}_{031}}\cdot\left(\frac{t}{\theta^{(\ell)}_{031}}\right)^{\theta^{(\ell)}_{032}-1},\ \ell=1,2,\end{equation}
where $\theta^{(\ell)}_{031}$ denotes the scale and $\theta^{(\ell)}_{032}$ the shape parameter, respectively. 
By assuming these two distributions, this scenario yields a very accurate approximation to the actual data, see Figure \ref{cumIntens} in Section \ref{appl}. More precisely, modeling the transition intensities by the Gompertz and the Weibull distribution, instead of assuming constant intensities, provides a much better initial model fit, resulting in a simulation setup with very realistic conditions with regard to the real data example. 

We choose the parameters given by the corresponding transition intensities of the application example (see Table \ref{tab:applparameter} in Section \ref{appl}), resulting in 
$$d=\max_{j=1}^3 \| \alpha^{(1)}_{0j} - \alpha^{(2)}_{0j} \|_\infty=\max\{0.0003,0.0028,0.0004\}=0.0028$$
for Scenario 3. 
Similar to Scenario 2, we obtain Scenario 4 in this setting by considering two identical models, such that 
$\theta^{(2)}=\theta^{(1)}$
and consequently we have $d=0$ in this case. Of note, Scenario 2 and Scenario 4 can only be used to simulate the power of the test. 
Table \ref{tab1} gives an overview of the simulation scenarios.


\small
\begin{table}[]
\caption{\small Chosen distributions of the simulation scenarios, the resulting maximum distance between transition intensities $d$ and the similarity thresholds $\Delta$ under consideration. Numbers in bold correspond to simulations of type I errors. As $d=0$ in Scenario 2 and Scenario 4, respectively, we only simulate the power there. (Exp.= Exponential) }
\small 
\begin{tabular}{l|lll|l|l}
           & \multicolumn{3}{c|}{Distribution}       & \multicolumn{1}{c|}{\multirow{2}{*}{d}} & \multirow{2}{*}{Thresholds $\Delta$}      \\
           & State 1     & State 2     & State 3     & \multicolumn{1}{c|}{}                   &                                           \\ \cline{1-6} 
Scen. 1 & Exp. & Exp. & Exp. & 0.0006                                  & \textbf{0.0006}, 0.001, 0.0015              \\
Scen. 2 & Exp. & Exp. & Exp. & 0                                       &  0.001, 0.0015           \\
Scen. 3 & Gompertz    & Gompertz    & Weibull     & 0.0028                                  & \textbf{0.002}, \textbf{0.0028}, 0.004, 0.005, 0.007, 0.01 \\
Scen. 4 & Gompertz    & Gompertz    & Weibull     & 0                                       &  0.004, 0.005, 0.007, 0.01                   
\end{tabular}\label{tab1}
\end{table}\normalsize

Based on the application example, where $n_1=213$ and $n_2=482$ patients are observed in the first and second group, we consider a range of different sample sizes, i.e., $n=(n_1,n_2)=(200,200),(250,300),(300,300),(250,450),(300,500)$ and $(500,500).$
Also driven by the application example, we assume administrative censoring with a given follow-up period of $90$ days. Consequently, we consider two competing risk models, each with $j=3$ states over the time range $\mathcal{T}=[ 0,90 ]$. If there is no transition to one of the three states, an individual is administratively censored at these $90$ days. 

To additionally investigate the effect of different types of censoring we consider a second setting replacing the administrative censoring by random right-censoring, where censoring times are generated according to an exponential distribution. Here, the observed time for an individual is given by the minimum of the simulated censoring time and the event time, respectively. By varying the rate parameter of the exponential distribution we are able to investigate the effect of different amounts of censoring. Precisely, we consider different rate parameters between $0.0002$ and $0.01$, resulting in approximately $16\%$ up to $85\%$ of the individuals being censored (details for the particular scenarios are given when discussing the results in Section \ref{results}). 
For the sake of brevity, when investigating the effect of random censoring, we restrict ourselves to Scenarios 1 and 3 respectively, and three different sample sizes, that is $n=(n_1,n_2)=(200,200),(300,300)$ and $(500,500)$.

The data in all simulations is generated according to the algorithm described in \citet{beyersmann2009}.
All simulations have been run using R Version 4.3.0. The total number of simulation runs is $N=1000$ for each configuration and due to computational reasons the test is performed using $B=250$ bootstrap repetitions. 
The computation time using an Intel Core i7 CPU with 32GB RAM for one particular dataset with $B=250$ bootstrap repetitions is approximately $10$ seconds for scenarios 1 and 2 and varies between $3$min and $11$min for scenarios 3 and 4, depending on the sample size under consideration.

\subsection{Results}\label{results}
\subsubsection{Scenario 1}

In order to simulate the type I error and the power of the procedure described in Algorithm \ref{alg1}, we consider different similarity thresholds $\Delta$.
When simulating type I errors, we assume $\Delta=d$ in both scenarios under consideration, reflecting the situation on the margin of the null hypothesis.
Thus, in Scenario 1, we set $\Delta=0.0006$. \\
First, we consider administrative censoring as described above, that is, a fixed end point of the study at $\tau=90$ days. The first row of Figure \ref{scen1a} displays the type I error rates of the procedure proposed in Algorithm \ref{alg1} in dependence of the sample size, directly compared to the ones derived by the ``Individual method'' presented in \cite{binder2022similarity} (see also Section \ref{simapp}). 
We observe that type I errors are much closer to the desired level of $\alpha=0.05$, whereas they are practically $0$ for the individual method. The still rather conservative behaviour of the test can be explained theoretically: according to Theorem \ref{thmneu} we expect type I errors to be smaller than $\alpha$, as transition intensities are constant and consequently their differences are constant functions as well, meaning that the set of points maximizing these functions each consists of the entire time range $\mathcal{T}$. 

The second and third row of Figure \ref{scen1a} visualize the power for both procedures, for $\Delta=0.001$ and $\Delta=0.0015$, respectively. For the latter the difference between the two methods is rather small and only visible for small sample sizes.
However, for $\Delta=0.001$ we clearly observe that the power of the new method is higher than the power of the individual method, for all sample sizes under consideration. This superiority is stronger for smaller sample sizes, which is in line with the theoretical findings, stating that both procedures are consistent as stated in \eqref{consistence2}.
\\

\begin{figure}[t]
\begin{center}
\includegraphics[width=0.85\textwidth]{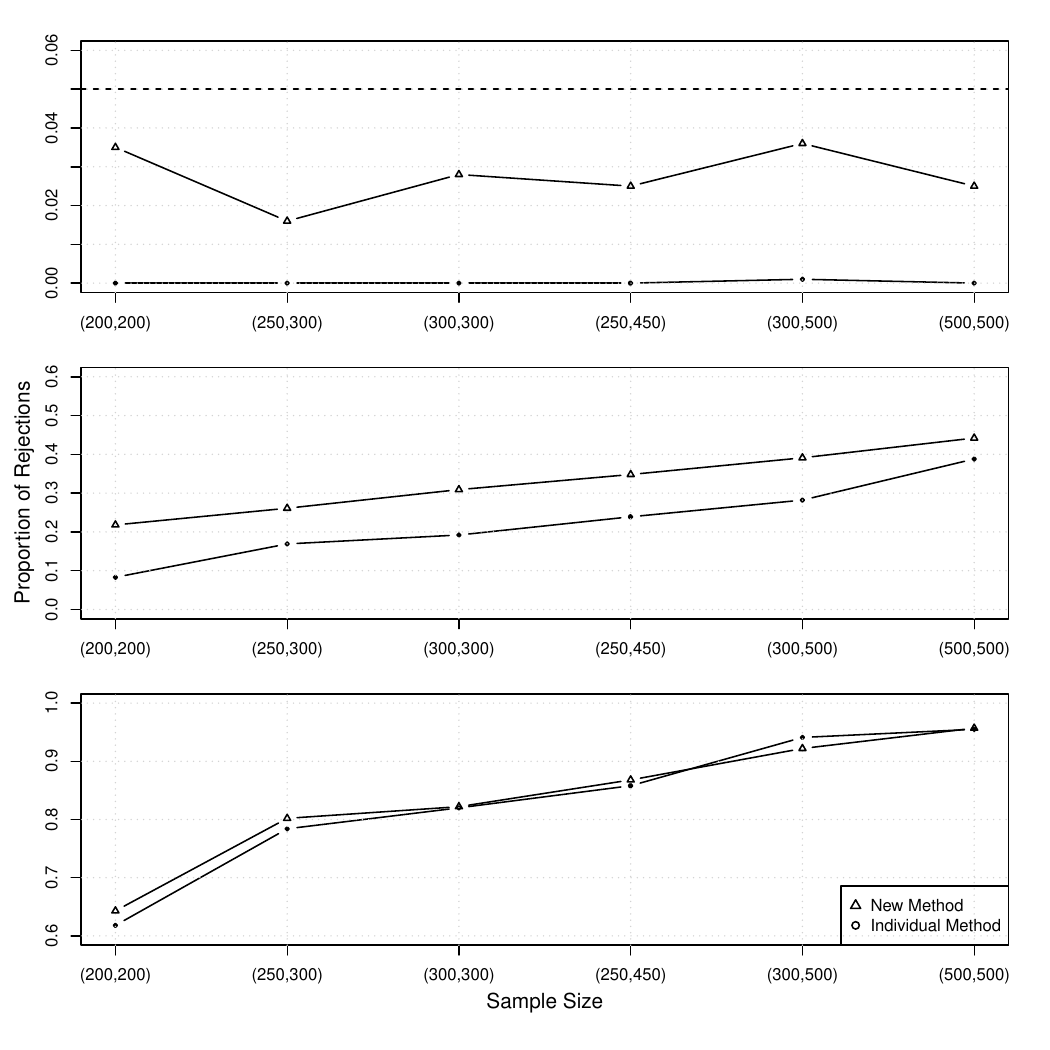}
\caption{Scenario 1: Proportion of rejections in dependence of the sample size for the new method and the individual method \cite{binder2022similarity}. The three rows display different choices of $\Delta$, that is $\Delta=0.0006$ corresponding to the null hypothesis in the top row, $\Delta=0.001$ in the middle and $\Delta=0.0015$ in the bottom row, where the latter two correspond to the situation under alternative. The dashed line in the first row indicates the nominal level chosen as $\alpha=0.05$.}
\label{scen1a}
\end{center}
\end{figure}

Regarding the effect of random right censoring, Table \ref{tab_cens} displays the results for the simulated type I error and the power considering different amounts of censoring. Precisely, we consider censoring rates between $0.001$ and $0.1$, resulting in approximately $22\%$ to $80\%$ of the individuals being censored. The first column corresponds to the null hypothesis \eqref{h2}, whereas the last two columns present the power of the procedures for the two different thresholds $\Delta=0.001$ and $\Delta=0.0015$, respectively. The numbers in brackets correspond to the results from the individual procedure \cite{binder2022similarity} for an easier comparability. 
It turns out that, in contrast to administrative censoring, the new method suffers from some type I error inflation for low sample sizes if censoring rates become large. The opposite holds for the individual procedure which is extremely conservative, as the simulated level is practically zero in all configurations. This type I error inflation disappears for increasing sample sizes. For instance, considering $n_1=n_2=500$ all simulated type I errors are below $0.06$, except for a censoring rate of $0.01$, corresponding to approximately $80\%$ of the individuals being censored. 
Hence we conclude that type I errors still converge to the desired level of $\alpha=0.05$ with increasing sample sizes.

Regarding the power we observe a substantial improvement with the new method for almost all configurations, particularly in case of small sample sizes and large censoring rates, e.g. achieving now a simulated power of $0.290$ instead of $0.076$ for $n_1=n_2=200$ and a censoring rate of $0.01$.
If sample sizes are large, the results of both procedures are qualitatively the same which is in line with the asymptotic theory stated in Binder et al.\cite{binder2022similarity} and in Theorem \ref{thmneu} of this paper. 

\small
\begin{table}[htb]
\caption{Scenario 1: Simulated level (column 4) and power (columns 5-6) of the new method, i.e. the test  described in Algorithm \ref{alg1}, considering different sample sizes, censoring rates and thresholds $\Delta$. The numbers in brackets correspond to the results from the individual procedure. The nominal level is chosen as $\alpha=0.05$. The third column displays the mean proportions of censored individuals.} \label{tab_cens}
	\centering
\begin{tabular}{c|c|c|c|c|c}
\hline
$(n_1,n_2)$ & Censoring rate & Censored & $\Delta=0.0006$ & $\Delta=0.001$ & $\Delta=0.0015$ \\
\hline
\multirow{5}{*}{(200,200)}      &   0.001  & $25\%$  & 0.037 (0.000) & 0.534 (0.476)   & 0.964 (0.852)      \\
                                &    0.002   & $40\%$  & 0.042 (0.000) & 0.410 (0.363)    & 0.916 (0.752)    \\
                                &    0.003  & $50\%$   & 0.053 (0.000) & 0.373 (0.313)   & 0.807 (0.656)    \\
                                &   0.005   & $63\%$   & 0.107 (0.000)  & 0.326 (0.226)   & 0.772 (0.480)      \\
                                &   0.01    & $77\%$   & 0.220  (0.000)  & 0.290 (0.076)  & 0.535 (0.209)    \\
                                \hline

\multirow{5}{*}{(300,300)}      & 0.001    & $25\%$  & 0.045 (0.001)  & 0.694 (0.618)    &0.990 (0.966)      \\
                                & 0.002   &  $40\%$ & 0.060 (0.000)  & 0.587 (0.553)     &0.965 (0.932)      \\
                                & 0.003  &  $50\%$ & 0.056 (0.000)  & 0.504 (0.470)    &0.902 (0.853)      \\
                                & 0.005   &  $63\%$    & 0.083 (0.001) & 0.359 (0.356)    &0.772 (0.758)      \\
                                & 0.01   &  $77\%$  & 0.161 (0.000)     & 0.307 (0.183)     &0.562 (0.452)      \\
                                \hline
\multirow{5}{*}{(500,500)}      & 0.001   &   $25\%$  & 0.050 (0.001)   & 0.847 (0.831)     &1.000 (1.000)      \\
                                & 0.002   &   $40\%$    & 0.057 (0.000)   & 0.791 (0.751)     &0.987 (0.985)      \\
                                & 0.003   &   $50\%$   & 0.052 (0.000)   & 0.685 (0.682)     &0.967 (0.976)      \\
                                & 0.005   &   $63\%$    & 0.060 (0.000)   & 0.545 (0.583)     &0.887 (0.922)      \\
                                & 0.01   &  $77\%$    & 0.119 (0.000)     & 0.371 (0.333)    & 0.664 (0.772)    \\ 
\hline
\end{tabular}
\end{table}
\normalsize 

\subsubsection{Scenario 2}

We still assume constant intensities for all transitions, but choose two identical models, that is $\theta^{(2)}=\theta^{(1)}$, resulting in $d=0$. Consequently, we now simulate the maximum power of the test. 
Figure \ref{scen2} (a) displays a direct comparison of the method proposed in Algorithm \ref{alg1} and the individual method. We observe that for the smaller similarity threshold of $\Delta=0.001$ the power of the new method is higher for all sample sizes under consideration. Of note, this effect is much more visible for smaller sample sizes. For instance, considering $n_1=n_2=200$ the power is given by $0.652$ for the new method and $0.415$ for the individual method, respectively, whereas almost identical values ($0.982$ and $0.987$, resp.) are observed for the largest sample size of  $n_1=n_2=500$.
Considering $\Delta=0.0015$, the same conclusion can be drawn for larger similarity thresholds, as all values of the simulated power are qualitatively the same across the two methods.

\begin{figure}[t]
\begin{center}
(a)\includegraphics[width=0.46\textwidth]{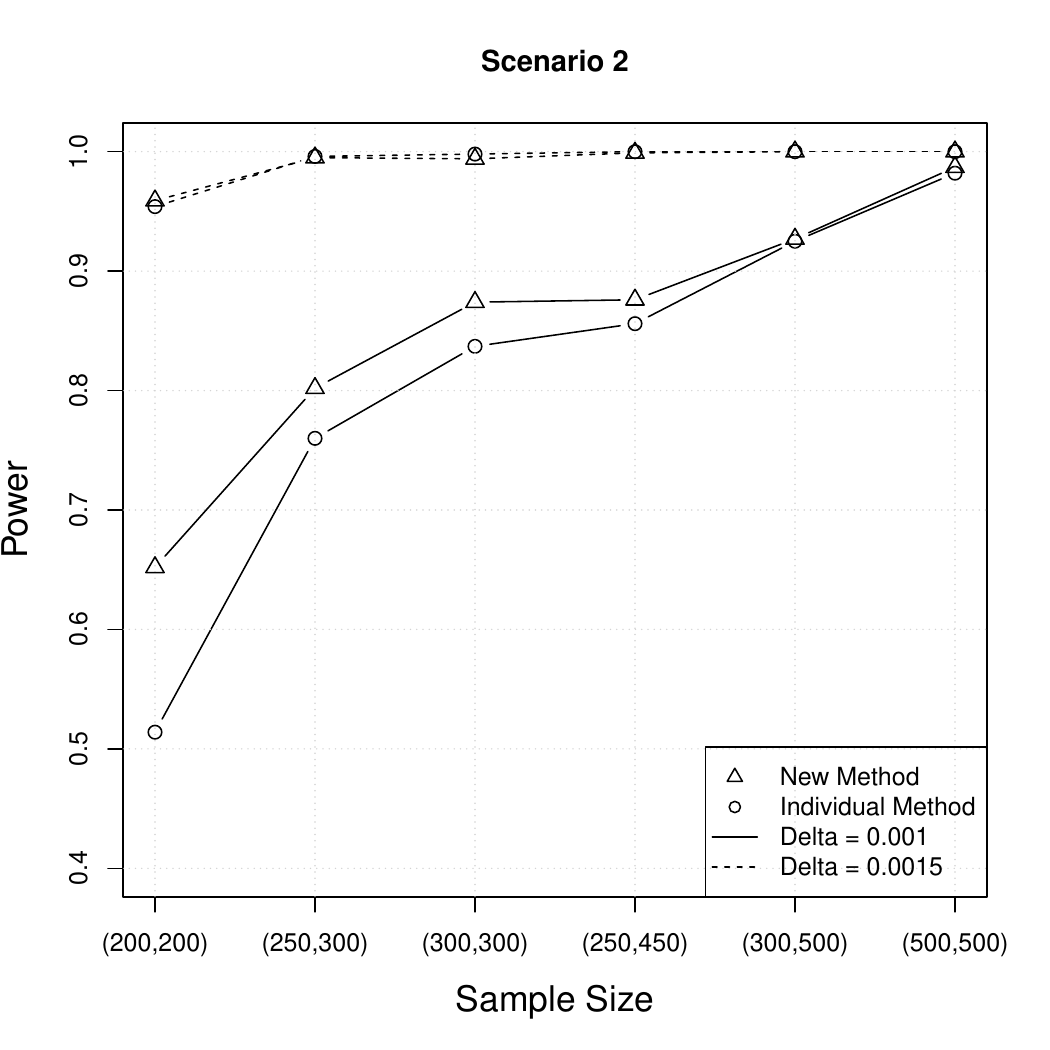}
(b)\includegraphics[width=0.46\textwidth]{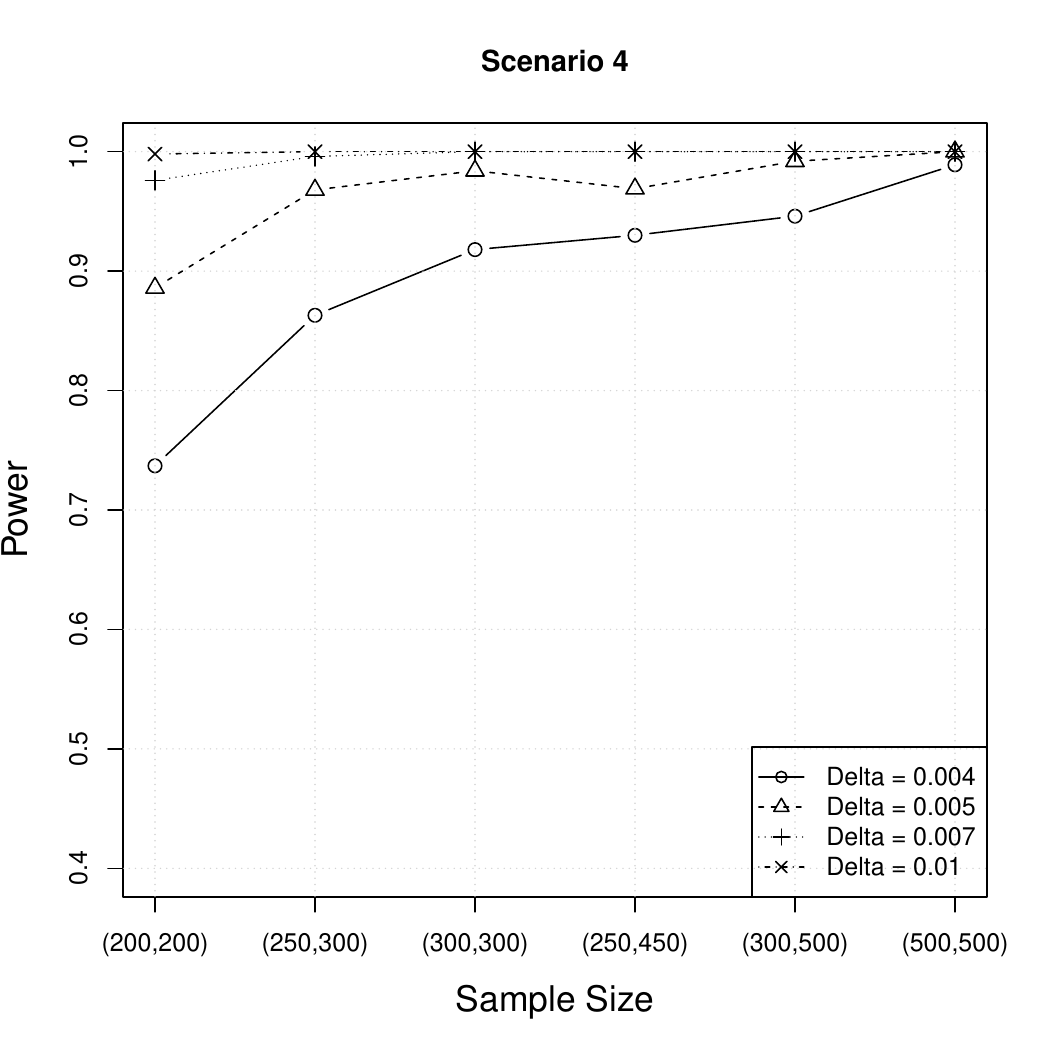}
\caption{(a) Scenario 2: Power of the new method and the individual method \cite{binder2022similarity} in dependence of the sample size for different similarity thresholds. (b) Scenario 4: Power of the new method in dependence of the sample size for different similarity thresholds. }
\label{scen2}
\end{center}
\end{figure}

\subsubsection{Scenario 3}
For simulating the type I error in Scenario 3 we consider $\Delta=0.002$ and $\Delta=0.0028$, the latter again reflecting the situation of being on the margin of the null hypothesis. 
 Note that for this choice of parameters for each of the three difference curves $\alpha^{(1)}_{0j}(t) - \alpha^{(2)}_{0j}(t)$, $j=1,2,3$, the maximum over the time range $\mathcal{T}$ is attained at one single point. Consequently, the set  $\mathcal{E}$ defined in Theorem \ref{thmneu} consists of this one point, meaning that this simulation scenario reflects the situation in \eqref{level2.1}. 

 Again, these theoretic findings are supported by the simulation results, which are displayed in Table \ref{tab:scen3}. Precisely, we observe that type I error rates converge to the desired level of $\alpha=0.05$ with increasing sample sizes. For instance, considering the scenario which is the closest to the application example, i.e. $(n_1,n_2)=(250,400)$ the simulated type I error is given by $0.059$. However, we observe a slight type I error inflation for the smaller samples under consideration, that is up to 300 patients per group. For example, the highest observed type I error is given by $0.110$, attained for sample sizes of $n_1=n_2=200$. Of note, for this configuration the number of expected transitions is only 36 for group 1 and 46 for group 2, respectively, due to the high amount of censoring (see also Section \ref{appl}).
The power increases with increasing sample sizes. We note that the threshold should not be too small, as the power is not very satisfying in this case. For instance, we observe a power of $0.2$ for a medium sample size of $n_1=n_2=300$ and a very small threshold of $\Delta=0.004$, whereas it almost doubles for $\Delta=0.005$ and finally approximates $1$ for $\Delta=0.01$. 

Finally, Figure \ref{scenario3cens} displays the results for the simulated type I error and the power considering different amounts of random right censoring for a fixed sample size of $n_1=n_2=500$. Censoring rates are chosen as $0.0002$, $0.001$, $0.002$ and $0.005$, resulting in mean proportions of censored individuals ranging from approximately $16\%$ up to $80\%$. We observe that even for high censoring rates the power is reasonably high and, moreover, higher than in case of administrative censoring at the end of the study. However, this comes at the cost of a slightly inflated type I error, which attains its maximum of $0.091$ for the highest censoring rate of $0.005$. When considering administrative censoring, which results in very similar proportions of censored individuals, the corresponding type I error is given by $0.055$, demonstrating that for this type of censoring the problem of type I error inflation does not occur. 

\begin{table}[htb]
\caption{Scenario 3: Simulated level (second column) and power (columns 3-6) of the new method, i.e. the test described in Algorithm \ref{alg1}, considering different sample sizes and thresholds $\Delta$. The nominal level is chosen as $\alpha=0.05$.} \label{tab_cens2}
	\centering
 \small
\begin{tabular}{c|c|c|c|c|c}
\hline
$(n_1,n_2)$ &  $\Delta=0.0028$ & $\Delta=0.004$ & $\Delta=0.005$ & $\Delta=0.007$ & $\Delta=0.01$ \\
\hline
\multirow{1}{*}{(200,200)} &      0.110     &   0.198    &  0.316   & 0.602 & 0.920     \\
\multirow{1}{*}{(250,300)} &      0.086    &   0.176    &   0.324  & 0.692 & 0.974    \\
\multirow{1}{*}{(300,300)} &      0.080     &   0.200        &  0.367    & 0.756 & 0.982   \\
\multirow{1}{*}{(250,450)} &   0.059    &  0.156   &   0.311    & 0.734  & 0.991  \\
\multirow{1}{*}{(300,500)} &   0.048    &   0.160    &  0.336  & 0.822 & 0.997   \\
\multirow{1}{*}{(500,500)} &  0.055  &   0.233     &  0.528   & 0.920 & 1.000     \\
\hline
\end{tabular}\label{tab:scen3}
\end{table}

\begin{figure}[t]
\begin{center}
\includegraphics[width=0.7\textwidth]{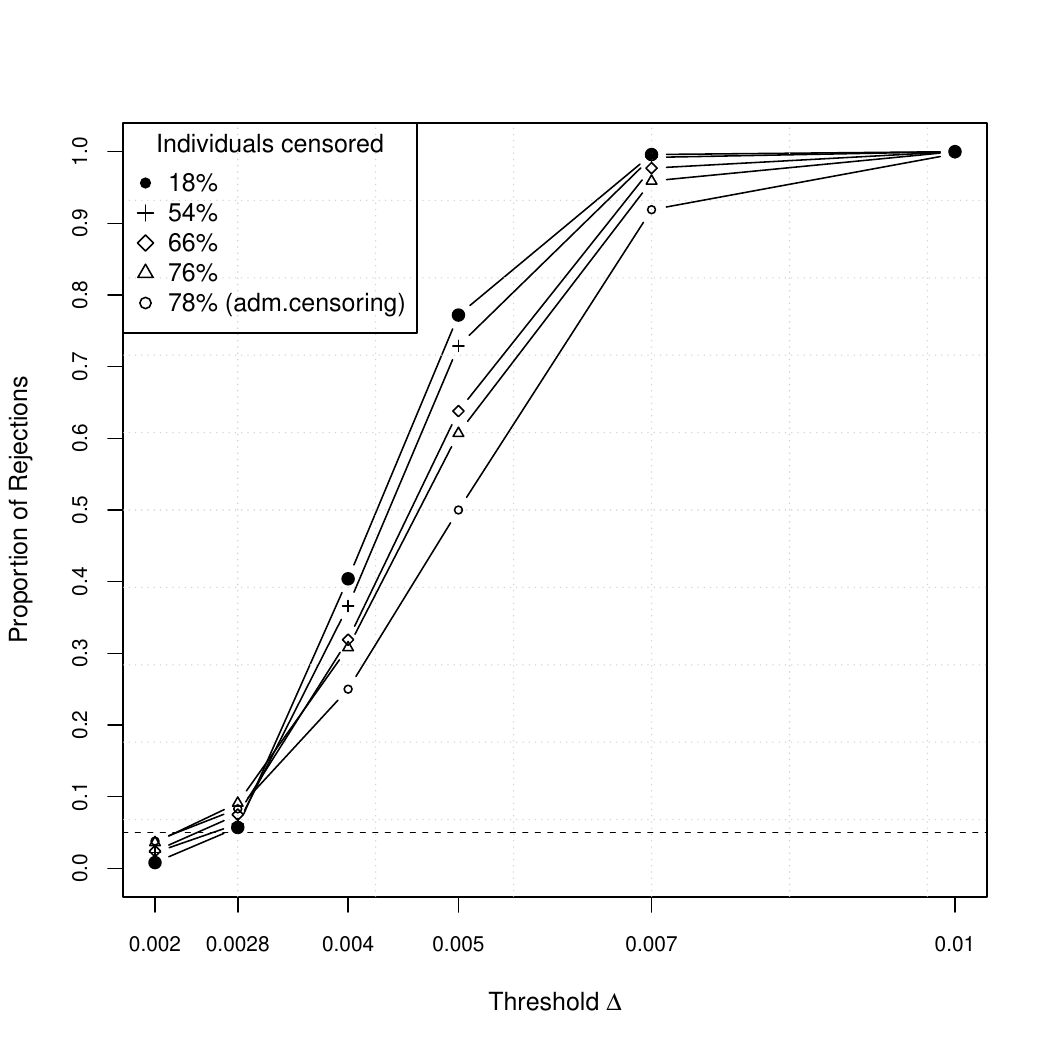}
\caption{Scenario 3: Proportion of rejections for different amounts of censoring at a fixed sample size of $n_1=n_2=500$ in dependence of the threshold. The first two thresholds correspond to the null hypothesis (where the second one displays the margin situation), the last four to the alternative. The dashed line indicates the nominal level $\alpha=0.05$. }\label{scenario3cens}
\end{center}
\end{figure}

\subsubsection{Scenario 4}

We now consider two identical models as in Scenario 2, but we assume a Gompertz distribution for the first two states and a Weibull distribution for the third one, respectively. All other configurations remain as described in Scenario 3. Consequently, we thereby simulate the maximum power, as $d=0$. Figure \ref{scen2} (b) displays the power of the test in dependence of the sample size for different similarity thresholds $\Delta$. We note that the power is reasonably high and above $0.8$ for all configurations except for the combination of the smallest threshold and the smallest sample size.

\section{Application example: Healthcare pathways of prostate cancer patients involving surgery} \label{appl}

\begin{figure}[t]
\begin{center}
\includegraphics[width=\textwidth]{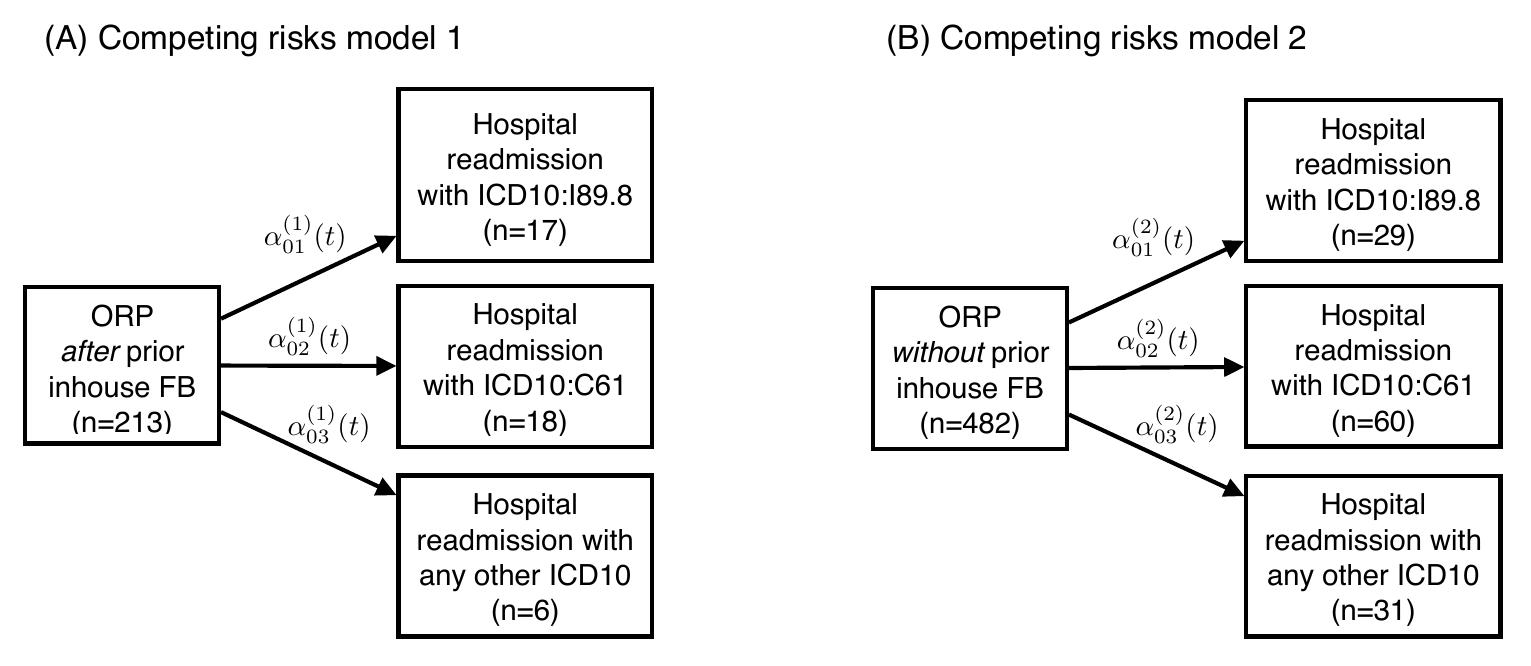}
\caption{Competing risks multi-state models illustrating healthcare pathways for two populations: (A) patients {\it receiving} inhouse fusion biopsy prior to open radical prostatectomy and (B) patients {\it not receiving} inhouse fusion biopsy prior to open radical prostatectomy. The arrows indicate the transitions between the states that are investigated. The $\alpha_{0j}^{(\ell)}(t)$, $ j=1,2,3,\ \ell=1,2$ mark the transition intensities as functions of time (see eq.~\eqref{intensities}).}
\label{uromsm}
\end{center}
\end{figure}

In our application example, we examine coding data from routine inpatient care of prostate cancer patients at the Department of Urology at the Medical Center - University of Freiburg, which was systematically processed as part of the German Medical Informatics Initiative. For each inpatient case, the main and secondary diagnoses are coded in the form of ICD10 codes (10th revision of the International Statistical Classification of Diseases and Related Health Problems); in addition, all applied and billing-relevant diagnostic and therapeutic procedures are coded together with a time stamp in the form of OPS codes (operation and procedure codes). 

\begin{figure}[t]
\begin{center}
\includegraphics[width=\textwidth]{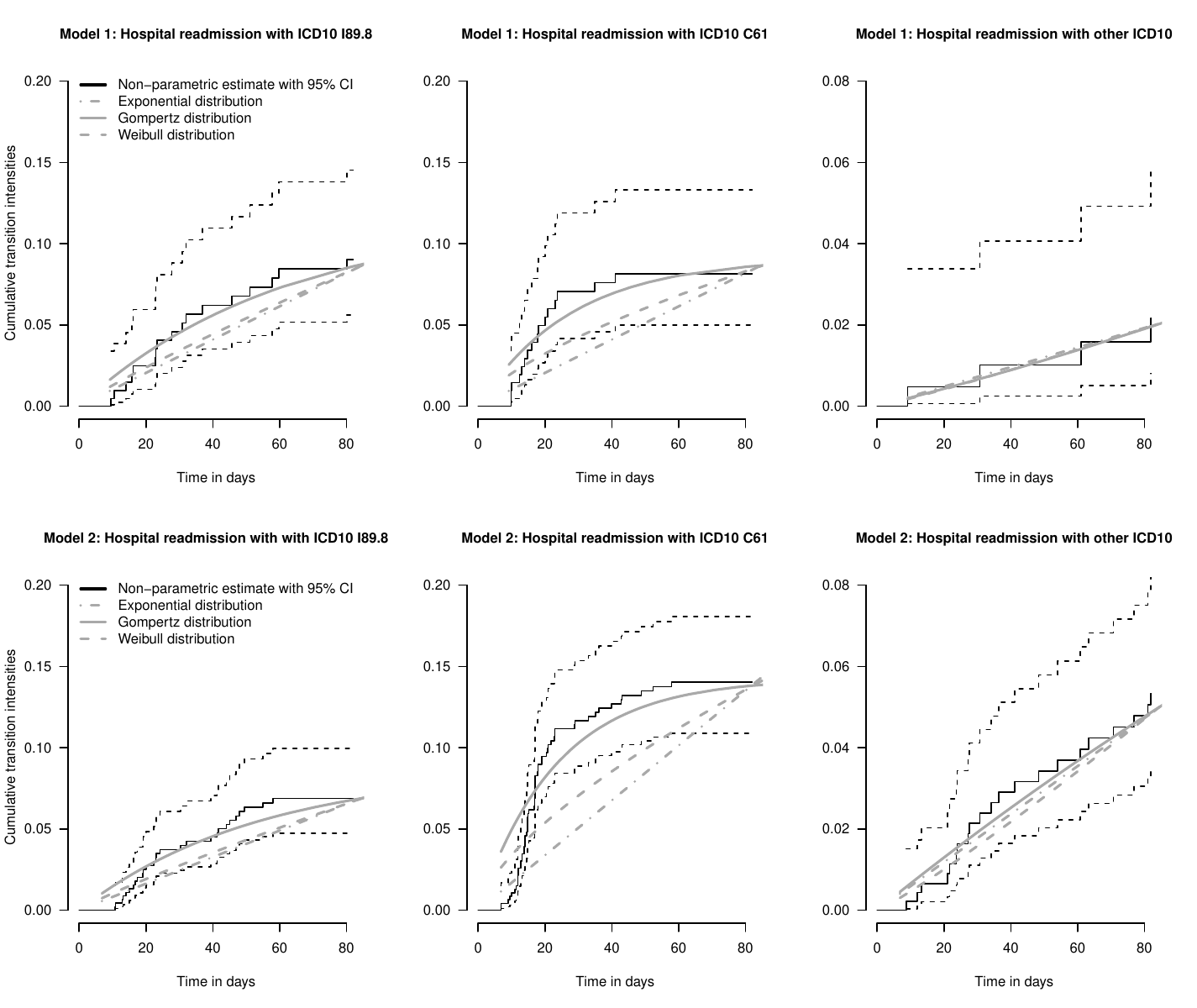}
\caption{Upper row: Estimates of the cumulative transition intensities from competing risks model 1 (upper row), and competing risks model 2 (lower row) in the application example. Illustrated for each panel are the nonparametric Nelson-Aalen estimates (black lines) along with 95\% confidence intervals (CI), as well as parametric model fits from Gompertz distribution (solid gray lines), Weibull distribution (dashed gray lines), and exponential distribution (dashed-dotted gray lines).}
\label{cumIntens}
\end{center}
\end{figure}

Specifically, we consider cases that have undergone open surgery with resection of the prostate including the vesicular glands, also known as open radical prostatectomy (ORP). We retrospectively identified all patients with prostate cancer who underwent ORP at the Department of Urology, University of Freiburg, between January 01, 2015 and February 01, 2021. This resulted in a total of n=695 patients. The current diagnostic standard before such a surgical procedure is a magnetic resonance imaging-based examination with targeted fusion biopsy (FB). In our data, n=213 (31\%) patients received an FB diagnosis prior to ORP, while a larger proportion of patients, n=482 (69\%), did not receive an FB diagnosis in the Department of Urology prior to ORP.

In the healthcare pathway after ORP, in some cases there are hospital readmissions due to competing causes, which can be attributed to the surgery in the period of typically 90 days after surgery. The question now is whether these pathways are similar irrespective of the type of prior diagnosis. Therefore, we distinguish two populations, $\ell = 1, 2$, based on the FB diagnosis obtained prior to surgery and aim to investigate the similarity of subsequent pathways using the two independent competing risk models, as shown in Figure~\ref{uromsm}, where the $\alpha_{0j}^{(\ell)}(t), j = 1, 2, 3, \ell = 1, 2$, describe the transition intensities to the different possible states in the model (see \eqref{intensities}). In the data, the following hospital readmissions occurred over time within 90 days after surgery: Lymphocele (ICD10:I89.9; Model 1: n=17, 8\%; Model 2: n=29, 6\%), malignant neoplasm of the prostate (ICD10:C61, Model 1: n=18, 8\%; Model 2: n=60, 12\%), or ``any other diagnosis'' (Model 1: n=6, 3\%; Model 2: n=31, 6\%). We administratively censor follow-up at 90 days after ORP. 

To understand the dynamics and magnitude of the different risks and to identify a suitable parametric distribution, we estimate the cumulative transition intensities in both models non-parametrically using the Nelson-Aalen estimator \cite{aalen_nonparametric_1978}. In addition, we fit an exponential, Weibull, and Gompertz model to the data. The estimates are shown in Figure~\ref{cumIntens}. For the first and second competing risks states in both models the estimates indicate a clear non-constant accumulated risk, and specifically the Gompertz distribution captures the time dynamics in all cumulative intensities best (as compared to the non-parametric estimates). For the third state a Weibull fit seems to be equally suitable as a fit from the Gompertz model, even the assumption of constant intensities seems to be met. As overall only few events were observed per state, the magnitude of the transitions intensities is low, and correspondingly the uncertainty of estimates relatively high. This is also reflected in the estimates of the parameters of the transitions intensities (see Table~\ref{tab:applparameter}).

\small
\begin{table}[]
\small 
\caption{\small Estimates of the parameters $\theta^{(\ell)}_{0j}$ \eqref{def:theta} of potential event time distributions for the three transition intensities from competing risks model 1 and competing risks model 2 in the application example. For Gompertz and Weibull, the first value corresponds to the scale and the second value to the shape parameter (following \eqref{scenario2a} and \eqref{scenario2b}). Numbers in bold are used in the simulation study.}
\begin{tabular}{l|p{1.6cm}p{1.6cm}p{1.6cm}|p{1.6cm}p{1.6cm}p{1.6cm}}
           & \multicolumn{3}{c|}{Model 1}  & \multicolumn{3}{c}{Model 2}      \\
           & $\hat\theta_{01}^{(1)}$ & $\hat\theta_{02}^{(1)}$   & $\hat\theta_{03}^{(1)}$  & $\hat\theta_{01}^{(2)}$ & $\hat\theta_{02}^{(2)}$   & $\hat\theta_{03}^{(2)}$ \\ \hline
 Exponential & \textbf{0.001} & \textbf{0.0011} & \textbf{0.004} & \textbf{0.0008} & \textbf{0.0017} & \textbf{0.0009} \\           
 Gompertz & \textbf{0.002,\newline -0.016} & \textbf{0.003,\newline -0.036} & 0.0002,\newline 0.003 &  \textbf{0.002,\newline -0.018} & \textbf{0.006,\newline -0.043} & 0.0007,\newline -0.003 \\
 Weibull & -0.112, 1304.5 & -0.38, 3098.3 & \textbf{0.097, 2894.8} & -0.12, 1729.8 & -0.404, 1595.9 & \textbf{0.108, 1242.1} \\
\end{tabular}\label{tab:applparameter}
\end{table}\normalsize

For investigating the similarity of the two competing risk models using Algorithm \ref{alg1}, we assume two different settings of event time distributions and various similarity thresholds $\Delta$, ranging from $0.0005$ to $0.0015$. Subsequently, when assuming constant intensities, we will compare the results of this analysis with the results obtained by the individual method \cite{binder2022similarity}.
Figure \ref{casestudy} displays the results of the tests in dependence of the similarity threshold $\Delta$. 
The p-values for the individual method are obtained by the maximum of the p-values of the three individual tests.
Figure \ref{casestudy}(a) directly yields a comparison of the two methods. As expected, the p-values of the test proposed in Algorithm \ref{alg1} are overall similar, but slightly smaller than the ones from the individual method. Consequently, according to the new method, the null hypothesis can be rejected for a threshold of $\Delta=0.0011$, but not using the individual method.
The p-values in Figure \ref{casestudy}(b) correspond to the more realistic setting of fitting  Weibull/Gompertz distributions. We observe that the threshold has to be at least $\Delta=0.005$ such that the null hypothesis can be rejected. Of note, as the difference of the curves lies on another scale as when assuming constant intensities, these results cannot be compared to the p-values displayed in Figure \ref{casestudy}(a).



\begin{figure}[t]
\begin{center}
(a)\includegraphics[width=0.46\textwidth]{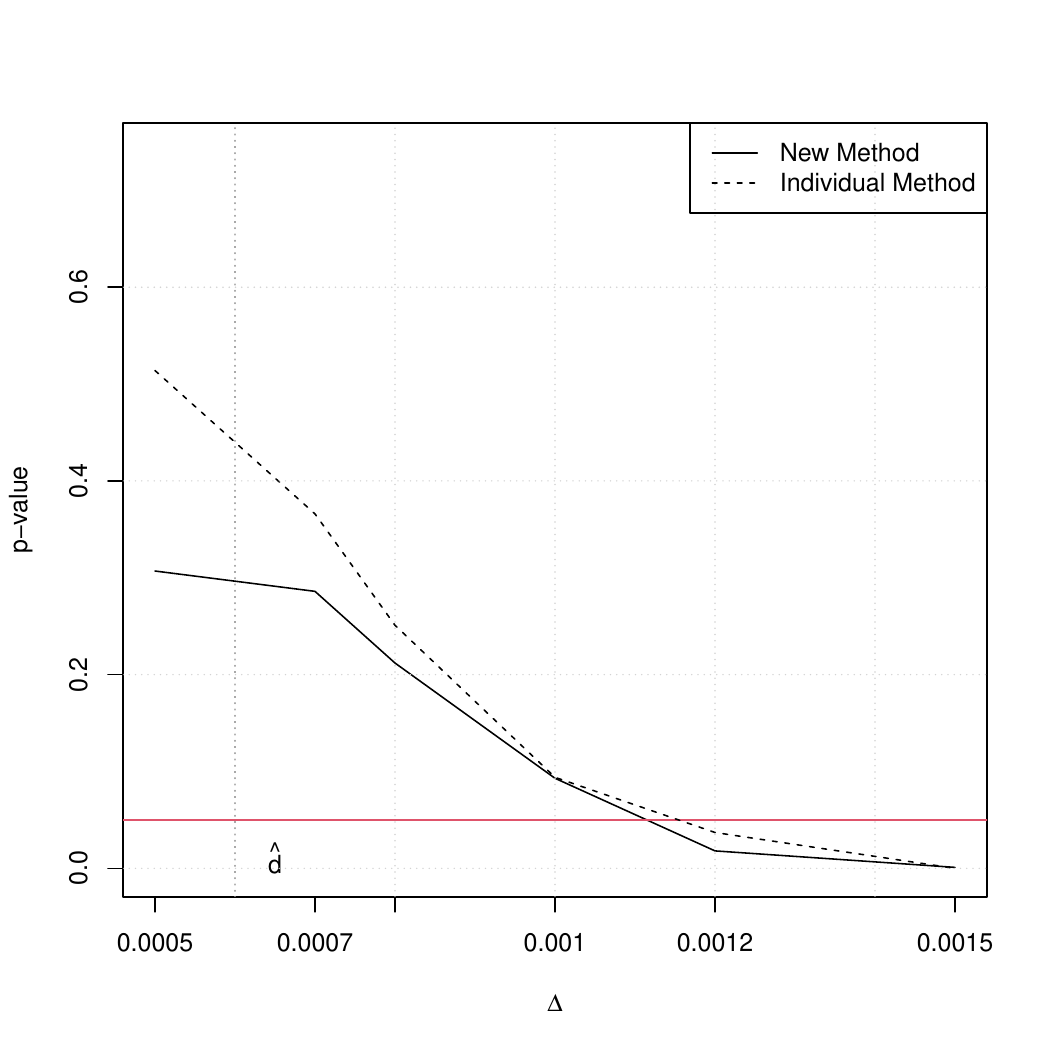}
(b)\includegraphics[width=0.46\textwidth]{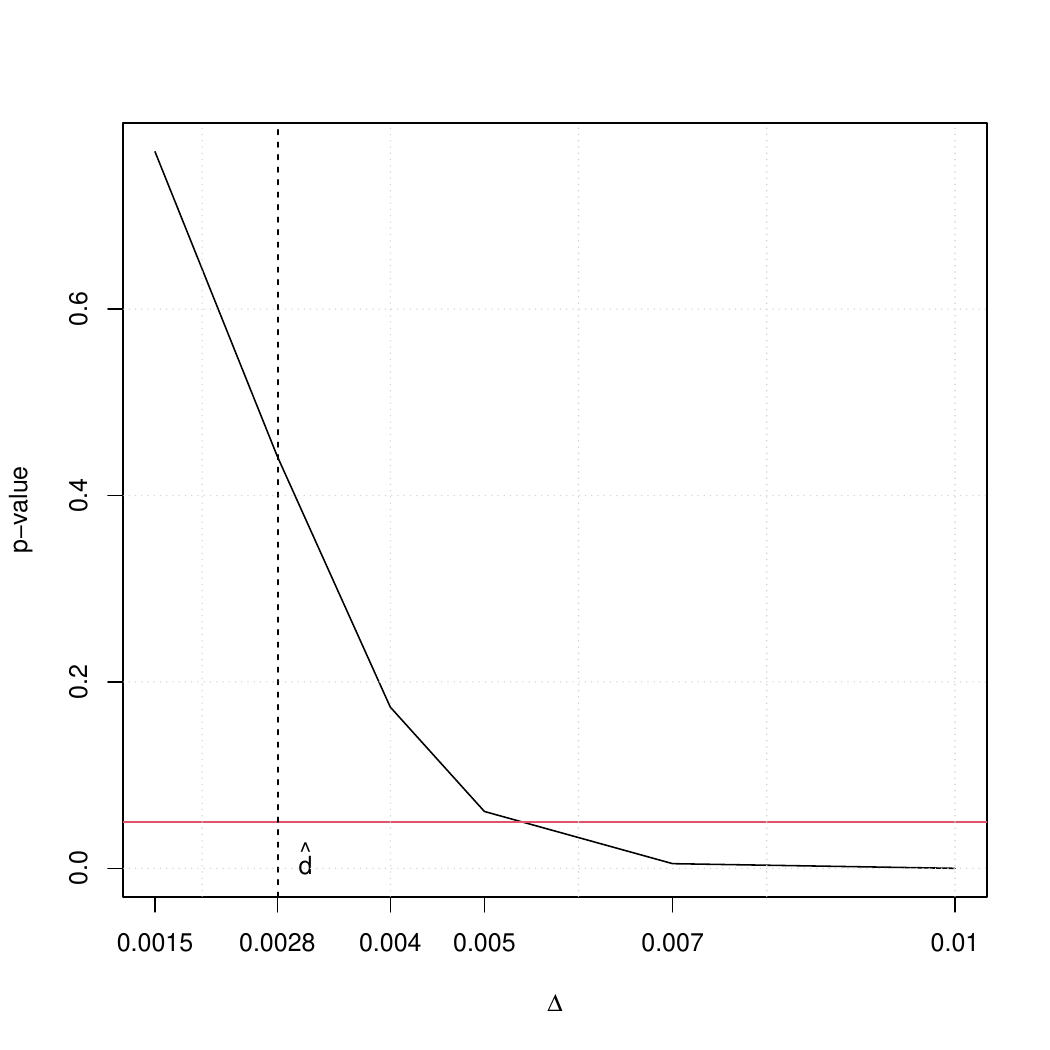}
\caption{(a) P-values of the test described in Algorithm \ref{alg1} (new method, solid line) compared to the individual method \cite{binder2022similarity} (dashed line) for the application example assuming constant intensities, in dependence of the threshold $\Delta$. (b) P-values of the test described in Algorithm \ref{alg1} assuming a Gompertz/Weibull model in dependence of the threshold $\Delta$. \\ The horizontal line indicates a p-value of 0.05, the vertical line indicates the test statistic $\hat d$.}
\label{casestudy}
\end{center}
\end{figure}


\section{Discussion} \label{disc}

In this work, we have addressed the question of whether two competing risk models can be considered similar. Building on the foundation laid by \cite{binder2022similarity}, we have extended the approach in two innovative ways.
First, we have successfully overcome the previous restriction to constant intensities. Our refined method introduces a framework that can incorporate arbitrary parametric distributions. This advance not only allows for a more nuanced modeling of transition intensities, but also leads to a more robust and effective testing procedure.
Second, we introduced a novel test statistic: the maximum of all maximum distances between transition intensities. This replaces the earlier method of aggregating individual state tests using the IUP. Through comprehensive simulation studies, we demonstrated the superior power of this new procedure.

While our approach introduces a unified similarity threshold $\Delta$, replacing the need for multiple individual thresholds $\Delta_j$, $j=1,\ldots,k$, this does come with a trade-off. The loss of detailed information in individual state comparisons is a consideration, but this is balanced by the increased overall power of the test. For those seeking detailed comparisons, individual tests should still be considered. However, for a broader assessment of the similarity between two competing risk models, our new approach is clearly superior.

An area for future exploration is the challenge of interpreting differences between transition intensities and establishing a meaningful similarity threshold. A potential solution could be to develop a test statistic based on ratios, allowing for more universally applicable thresholds, such as a permissible deviation of $10\%$. This would simplify the process, accommodating ratios within the range of $0.9$ to $1.1$, regardless of the absolute intensity values. We look forward to further research in this direction.

\section*{Funding information}

\noindent The work of N. Binder and H. Dette has been funded in part by the Deutsche Forschungsgemeinschaft (DFG, German Research Foundation) – Project-ID 499552394 – SFB Small Data.

\bibliographystyle{statinmed}

\newpage 
\section*{Appendix}

\begin{proof}[Proof of Theorem \ref{thmneu}]
	Recall the definition of the vector of parameters 
	$\theta^{(\ell)}\in \mathbb{R}^p$ in \eqref{hol2} 
	($\ell =1,2$) and define   $\theta  =  \big ((\theta^{(1)})^\top    ( \theta^{(2)})^\top \big )^\top  \in \mathbb{R}^{2p}$   as the vector of all parameters in the two competing risk models. Furthermore denote by  $ \hat \theta^{(\ell)}_{0j}$, $\hat \theta^{(\ell)} $
	and   $\hat\theta  =  \big (\hat\theta^{(1)})^\top    ( \hat\theta^{(2)})^\top \big )^\top $ the corresponding maximum likelihood estimators defined by maximizing \eqref{likelihood} (or equivalently \eqref{loglikelihood}) and
	define by 
	\begin{align} \label{hol1}
		\hat  \alpha_{0j}^{(\ell  )}(t)=\hat \alpha_{0j}^{(\ell  )}(t,  \hat  \theta^{(\ell)}_{0j})
	\end{align}
	the corresponding estimators of the transition intensity functions (note that 
	for $j=1, \ldots , k$, $\ell =1,2$ \eqref{hol1} defines   a $2k$-dimensional vector of functions defined on  the interval $\mathcal{T} = [0,\tau ]$). Then, by 
	Theorem  2  
	in \cite{borgan1984}   it follows that $\sqrt{n} (\hat  \theta -  \theta ) $ converges weakly to a multivariate normal distribution with mean vector $0$ and a block diagonal covariance matrix. 
	We  now interpret  the  vectors as stochastic processes on the finite set $ \mathcal{M} =
	\{1, \ldots  , k \} \times \{1,2 \} $ and rewrite this weak convergence 
	as 
	\begin{align} \label{hol11}
		\big  \{ \sqrt{n} (\hat  \theta^{(\ell)}_{0j}    -  \theta^{(\ell)}_{0j} )  \big  \}_{ (j, \ell) \in \mathcal{M} } 
		\rightsquigarrow  \big  \{ \mathbb{D} (j, \ell)   \big  \}_{(j, \ell) \in \mathcal{M}  }~.
	\end{align}
	Therefore, 
	an application of 
	the continuous mapping theorem \citep[see, for example,][]{vaart1998} 
	implies the weak convergence of the process
	\begin{align} 
		\label{hol12a}
		&  \big  \{ \sqrt{n} \big  (  (\alpha^{(1)}_{0j} (t,  \hat  \theta^{(1)}_{0j})   -  \alpha^{(1)}_{0j}  (t, \theta^{(1)}_{0j} ) )
		-  (\alpha^{(2)}_{0j} (t,  \hat  \theta^{(2)}_{0j})   -  \alpha^{(2)}_{0j}  (t, \theta^{(2)}_{0j} ) )
		\big
		)  \big  \}_{(j, ,t) \in \mathcal{X} }  \\
		\nonumber  
		& \qquad \qquad \qquad \qquad \qquad \qquad \qquad \qquad \qquad \qquad \qquad \qquad \qquad \qquad 
		\rightsquigarrow  \big  \{ \mathbb{G} (j,t)   \big  \}_{(j, t) \in \mathcal{X} }~
	\end{align}
	in $\ell^\infty (\mathcal{X})$,
	where $\mathcal{X}  = \{ 1 , \ldots ,  k \} \times \mathcal{T}$  and  $\{ \mathbb{G} (j, t)  \}_{(j, t) \in \mathcal{X}} $ is a centered 
	Gaussian process on $\mathcal{X}$. Note that this is the analog of the equation (A.7) in 
	\cite{detmolvolbre2015}, and it follows by similar arguments  as stated in this paper that 
	\begin{align} \label{hol12}
		\sqrt{n} \big ( \|  \hat  \alpha^{(1)} -  \hat  \alpha^{(2)}    \|_{\infty , \infty }    - \|  \alpha^{(1)} -  \hat  \alpha^{(2)}  
		\|_{\infty , \infty }   \big ) \rightarrow \max  \big\{ \max_{(j, t) \in \mathcal{E}^+} \mathbb{G} (j, t)
		, \max_{(j, t)\in \mathcal{E}^-} -\mathbb{G} (j, t) \big \}  ~, 
	\end{align}
	where the vectors $\hat  \alpha^{(\ell )}$ and $\hat  \alpha^{(\ell )}$ are defined by 
	$\hat  \alpha^{(\ell)} (t) = ( \hat \alpha^{(\ell)}_{0j} (t,  \hat  \theta^{(\ell)}_{0j})_{j \in \{ 1, \ldots , k \} }
	$ and 
	$  \alpha^{(\ell)} (t) = (  \alpha^{(\ell)}_{0j} (t,  \hat  \theta^{(\ell)}_{0j})_{j \in \{ 1, \ldots , k \} },
	$ respectively, 
	and
	$$
	\mathcal{E}^{\pm} = \Big \{ (j,t) \in \{ 1, \ldots , k \} \times \mathcal{T} ~:~ 
	\hat\alpha^{(1)}_{0j} (t) - \hat\alpha^{(2)}_{0j} (t)  = \pm 
	\| \hat\alpha^{(1)} - \hat\alpha^{(2)}\|_{\infty,\infty} 
	\Big \} ~.
	$$
	Note that $ \mathcal{E}^-\cup  \mathcal{E}^+ =  \mathcal{E}$, where $ \mathcal{E}$ is defined in \eqref{det1},
	and that 
	\eqref{hol12} is the analog of Theorem 3 in \cite{detmolvolbre2015}. Similarly, we obtain the weak convergence of the bootstrap process and the corresponding 
	statistic, that is 
	\begin{align} 
		\label{hol13}
		&  \big  \{ \sqrt{n} \big  (  (\alpha^{(1)}_{0j} (t,  \hat  \theta^{*(1)}_{0j})   -  \alpha^{(1)}_{0j}  (t, \hat{\hat  \theta}^{(1)}_{0j} ) )
			-  (\alpha^{(2)}_{0j} (t,  \hat  \theta^{*(2)}_{0j})   -  \alpha^{(2)}_{0j}  (t, \hat{\hat \theta}^{(2)}_{0j} ) )
			\big
			)  \big  \}_{(j, ,t) \in \mathcal{X} } 
		\\
		\nonumber  
		& \qquad \qquad \qquad \qquad \qquad \qquad \qquad \qquad \qquad \qquad \qquad \qquad \qquad \qquad 
		\rightsquigarrow  \big  \{ \mathbb{G} (j,t)   \big  \}_{(j, t) \in \mathcal{X} }~
	\end{align}
	and
	\begin{equation*} 
		\sqrt{n} \big ( \|  \hat  \alpha^{*(1)} -  \hat  \alpha^{*(2)}    \|_{\infty , \infty }    - \| \hat  \alpha^{(1)} -  \hat  \alpha^{(2)} 
		\|_{\infty , \infty }   \big ) \rightarrow \max  \big\{ \max_{(j,t) \in \mathcal{E}^+} \mathbb{G} (j,t)
		, \max_{(j,t)\in \mathcal{E}^-} -\mathbb{G} (j,t) \big \}  ~, 
	\end{equation*}
	conditionally on $X^{(1)}_{1} , \ldots , X^{(1)}_{n_1}, X^{(2)}_{1} , \ldots , X^{(2)}_{n_2}  $,
	where $\hat \alpha^{*(\ell)}$ is the bootstrap version  of $\hat \alpha^{(\ell)}$ and  
		$\hat{\hat \alpha}^{(\ell) } $ is obtained by the constrained estimates $\hat{\hat \theta}_{0j}^{(\ell) } $, i.e. $\hat{\hat{\alpha}}_{0j}^{(\ell  )}(t)=\alpha_{0j}^{(\ell  )}(t, \hat{\hat{\theta}}^{(\ell)}_{0j})$, $j=1,\ldots,k$, $\ell=1,2$, see also Algorithm \ref{alg1}, step (2).
	This  is the analog 
	of statement (A.25) in \cite{detmolvolbre2015}. 
	Now the statements (A.7) and (A.25) and their Theorem 3 are the main ingredients for the proof of Theorem 4 in \citet{detmolvolbre2015}. In the present context these statements can be replaced by  \eqref{hol12a},  \eqref{hol13} and \eqref{hol12}, respectively, and a careful inspection of the arguments given in \citet{detmolvolbre2015} proves the claim of Theorem \ref{thmneu}.
\end{proof}

\end{document}